\documentclass[12pt,amstex,leqno]{amsart}

\usepackage{latexsym}
\usepackage{amsthm}
\usepackage{amsmath}

\usepackage{amsfonts}
\usepackage{amssymb}
\usepackage{appendix}
\usepackage[mathcal]{euscript}
\usepackage{mathrsfs} 

\newcommand\op{\operatorname{op}}
\newcommand\card{\operatorname{card\,}}
\newcommand\Set{\operatorname{\bf Set}}

\newcommand{\id}{\operatorname{id}}
\newcommand{\Id}{\operatorname{Id}}

\newcommand\colim{\operatorname{\it colim}}
\newcommand{\Ord}{\operatorname{Ord}}

\newcommand\cp{\mathcal {P}}

\newcommand\N{\mathbb{N}}

\input xy
\xyoption{all}

\swapnumbers

\newtheorem{theorem}{Theorem}[section]
\newtheorem{lemma}[theorem]{Lemma}

\newtheorem{prop}[theorem]{Proposition}
\newtheorem{corollary}[theorem]{Corollary}

\theoremstyle{definition}
\newtheorem{definition}[theorem]{Definition}

\newtheorem{example}[theorem]{Example}
\newtheorem{exam}[theorem]{Example}
\newtheorem{exs}[theorem]{Examples}
\newtheorem{remark}[theorem]{Remark}

\newtheorem{notation}[theorem]{Notation} 
\newtheorem{assumpt}[theorem]{Assumption}

\numberwithin{equation}{section}

\begin{document}

\title{On Free Completely Iterative Algebras}

\author[Ji{\v{r}}{\'{\i}} Ad{\'{a}}mek]{Ji{\v{r}}{\'{\i}} Ad{\'{a}}mek$^*$}

\address{\newline
Department of Mathematics,
Faculty of Electrical Engineering,
 \newline Czech Technical University in Prague, \newline Czech Republic} 
\email{j.adamek@tu-bs.de}

\subjclass{ }

\begin{abstract}

For every finitary set functor $F$ we demonstrate that  free algebras carry a canonical partial order.
In case $F$ is bicontinuous, we prove that the cpo obtained as the   conservative  completion of the free algebra is the free completely iterative algebra. Moreover, the algebra structure of the latter is the unique continuous extension of the  algebra structure of the free algebra.

For general finitary functors the free algebra and the free completely iterative algebra are proved to be posets sharing  the same  conservative completion. And for every recursive equation $e$ in the free completely iterative algebra we present an $\omega$-chain of approximate solutions in the free algebra whose join is the solution of $e$.

\end{abstract}

\maketitle

\section{Introduction}\label{sec1}

\footnote{$^\ast$Supported by the Grant Agency of the Czech Republic under the grant 19-00902S.}

Recursion and iteration belong to the crucial concepts of theoretical computer science. An algebraic treatement was suggested by Elgot who introduced  iterative algebraic theories in \cite{E}. The corresponding concept for algebras over a given endofunctor $F$ was defined  by Milius \cite{M}: an algebra is called completely iterative  if every recursive equation  has a unique solution in it. We recall this in Section~5.
The free completely iterative theory of Elgot is then precisely
the algebraic theory corresponding to the free completely iterative algebras.
 Milius also described the free completely iterative algebra on a given object $X$: it is precisely the terminal coalgebra for the  endofunctor $F(-) + X$. This corresponds nicely to the fact  that the free  algebra on $X$ is precisely the initial algebra for $F(-) + X$.

In the present paper we study iterative  algebras for a finitary set functor $F$ (i.e., one preserving filtered colimits). We first  show that given a choice of  an element of $F\emptyset$, we obtain a canonical partial order  on the initial algebra $\mu F$  and on the terminal coalgebra $\nu F$. To illustrate this, consider the polynomial functor $H_\Sigma$ for a finitary signature $\Sigma$: here $\nu H_\Sigma$ is the algebra of all $\Sigma$-trees and $\mu H_\Sigma$ the subalgebra of all finite  $\Sigma$-trees. The ordering of $\nu H_\Sigma$ is `by cutting': for two $\Sigma$-trees $s$ and $s'$ we put $s<s'$ if $s$ is obtained from $s'$ by cutting, for a certain height, all nodes of larger heights away. This makes $\nu H_\Sigma$ a cpo which is 
the conservative completion  of the subposet $\mu H_\Sigma$. (The basic reason is that for every infinite $\Sigma$-tree  its cuttings $\partial _n s$ at level $n\in \N$ form an $\omega$-chain with $s= \sqcup \partial_n s$.) Now every finitary set functor can be presented as a quotient of a polynomial functor, see Section~4, and both  $\mu F$ and $\nu F$ inherit their orders from the order of $\Sigma$-trees by cutting. We prove that
\begin{enumerate}
\item[(a)] if $F$ is bicontinuous, i.e., it also preserves  limits  of $\omega^{\op}$-sequences, then $\nu F$ is a cpo which is the  conservative completion (see Remark \ref{R:ideal}) of $\mu F$, and
\item[(b)]
for finitary set functors in general $\nu F$ and $\mu F$ share the same conservative completion.
\end{enumerate}
Moreover, the coalgebra structure of  $\nu F$ is the unique continuous extension of the inverted algebra structure of $\mu F$. 
And for every coalgebra $A$ the unique homomorphism into $\nu F$ is a join of an $\omega$-chain of approximate homomorphisms $h_n\colon A \to \mu F$.
All this depends  on the choice of an element in $F\emptyset$.

We then apply this to a new description of the free completely iterative algebra on an arbitrary set $X\ne \emptyset$. We choose a variable in $X$ and obtain an order on $\Phi X$, the free algebra for $F$ on $X$, and one on $\Psi X$, the free completely iterative  algebra on $X$. We  prove that the  conservative completion of $\Phi X$ and $\Psi X$ coincide. And that  in case that $F$ is bicontinuous, $\Psi X$ is the conservative completion  of $\Phi X$. In both cases, the algebra structure of $\Psi X$ is the unique continuous  extension of that of $\Phi X$.
Moreover, solutions of recursive equations in $\Psi X$ can be obtained as joins of $\omega$-chains of so-called approximate solutions in $\Phi X$  obtained in a canonical manner.

\vskip 2mm
\textbf{Related Work.} 
We can work with complete metrics in place of complete partial orders. Barr proved that given a bicontinuous set functor $F$ with $F\emptyset \ne\emptyset$, there is a canonical complete metric on $\nu F$ which is the Cauchy completion of $\mu F$, see \cite{B}. This was extended in \cite{A2} to finitary set functors with $F\emptyset \ne \emptyset \colon \nu F$ and $\mu F$ have the same Cauchy completion, and the coalgebra structure of $\nu F$ is the unique continuous extension of the inverted algebra structure  of $\mu F$.

In the bicontinuous case a cpo structure of $\nu F$ was presented in \cite{A}. But the definition was quite technical; we recall this in Section~\ref{sec3}. One of the main results of the present paper that the order of $\nu F$ by  cutting (inherited from $\Sigma$-trees) coincides with that of op. cit.

\section{Polynomial Functors}\label{sec2}

We first illustrate our method on the special case: the \textit{polynomial functor} $H_\Sigma$ associated with a signature $\Sigma = (\Sigma_n)_{n\in \N}$. This is a set functor  given by
$$
H_\Sigma X =\coprod_{n\in \N} \Sigma_n \times X^n\,,
$$
and we represent the elements of the above set as `flat' terms $\sigma(x_1,\dots , x_n)$ where $\sigma \in \Sigma_n$ and $(x_i)\in X^n$.

\begin{remark}\label{R:psi} 
(1) A free algebra $\Phi_\Sigma X$ on a set is the algebra of all terms with variables in $X$. This can be represented  by finite trees as follows. A \textit{$\Sigma$-tree} is an ordered tree labelled in $\Sigma$ so that every node labelled in $\Sigma_n$ has precisely $n$ successors.
We consider $\Sigma$-trees up to isomorphism. Now given a set $X$  we form a new signature
$$
\Sigma_X= \Sigma+X
$$
 in which elements of $X$ have arity $0$.
A $\Sigma_X$-tree is called a \textit{$\Sigma$-tree over $X$}; its leaves are labelled by nullary symbols or variables from $X$.
 Then we get
 $$
 \Phi_\Sigma X = \ \mbox{all finite\ \ $\Sigma $-trees over $X$.}
$$
The algebra structure 
$$ 
\varphi \colon  H_\Sigma \big(\Phi_\Sigma X\big) \to \Phi_\Sigma X
$$
assigns to each  member $\sigma(t_1, \dots, t_n)$ (where $t_i$ are finite $\Sigma_X$-trees)
the $\Sigma_X$-tree with root labelled by $\sigma$ and with $n$ maximum proper subtrees $t_1,\dots , t_n$. Thus $\varphi^{-1}$ is tree tupling.

\vskip  1mm
(2)
The terminal coalgebra $\nu H_\Sigma$ can  analogously be  described  as the coalgebra of \textit{all\/} $\Sigma$-trees, the coalgebra operation is tree-tupling. For every set $X$ we denote by $\Psi_X$ the terminal  coalgebra of $H_{\Sigma_X} $ \ ($= H_\Sigma (-) + X$):
$$
\Psi_\Sigma X = \nu H_{\Sigma_X} = \nu (H_\Sigma +X)\,.
$$
It consists of all $\Sigma$-trees over $X$.
The coalgebra structure 
$$
\tau \colon \Psi_\Sigma X \to H_\Sigma (\Psi_\Sigma X)
$$
 assigns to a tree $t\in \Psi_\Sigma X$ either $x\in X$, if $t$ is a root-only tree labelled in $X$, or  $\sigma(t_1, \dots, t_n)$, if the root of $t$ is  labelled by $\sigma\in \Sigma_n$ and its successor subtrees are $t_1, \dots , t_n$. This is a free completely iterative algebra for $H_\Sigma$, see  Section~5.

\end{remark}

\begin{exam}\label{E:psi}
(1) If $\Sigma$ consists of a set $A$ of unary operation symbols, we have $H_\Sigma X=A\times X$. A tree in $\Psi_\Sigma X$ is either a finite unary tree over $X$ corresponding to an element of $A^\ast \times X$ (a leaf labelled in $X$, the other nodes labelled in $A$) or an infinite unary tree corresponding to a word in $A^\omega$:
$$
\Psi_\Sigma X = A^\ast \times X + A^\omega\,.
$$

(2) Let $\Sigma$ be a signature of one $n$-ary symbol for every $n\in \N$. Thus $H_\Sigma X = X^\ast$. A tree in $\Psi _\Sigma X$ does not need labels for inner nodes, and for leaves  we either have a label in $X$ or we consider the leaf unlabelled:
\begin{align*}
\Psi_\Sigma X =& \ \mbox{all finitely branching trees with leaves}\\
& \ \mbox{partially labelled in $X$}.
\end{align*}
\end{exam}

\begin{notation}\label{N:cut}
Let us choose an element $p\in X \cup \Sigma_0$. Then every tree $t$ in $\Psi_\Sigma X$ yields a tree $\partial_n t$ of height at most  $n$ by cutting all nodes of larger heights away and relabelling all leaves of height $n$ by $p$.
\end{notation}

\begin{definition}\label{D:cut}
We consider $\Psi_\Sigma X$ as a poset where for distinct trees $s$, $s'$ we put
$$
s<s' \quad \mbox{iff $s$ is a cutting of $s'$.}
$$

That is, $s= \partial_n s'$ for some $n\in \N$.
\end{definition}

\begin{example}\label{E:cut}

(1) For $H_\Sigma X= A\times X$ the subset $A^\omega$ of $\Psi_\Sigma X$ is discretely ordered. Given $(u, x)$ and $(v, y)$ in $A^\ast \times X$ then 
$$
(u,x) < (v,y) \quad \mbox{iff $u$ is a proper prefix  of $v$ and $x=p$.}
$$
Finally $(u,x)<w$, for $w\in A^\omega$, iff $u$ is a finite prefix of $w$ and $x=p$.

\vskip 2mm
(2) For $H_\Sigma X =X^\ast$ the set $\Psi_\Sigma X$ is ordered by cutting.
\end{example}

\begin{remark}\label{R:cut}
(a) Every tree $s$ in $\Psi_\Sigma X$ is a join of its cuttings:
$$
s=\bigsqcup
_{n\in \N} \partial _n s\,.
$$

(b) Every strictly increasing sequence $(s_n)_{n\in \N}$ in $\Psi_\Sigma X$ lies in $\Phi_\Sigma X$, i.e., each $s_n$ is finite. And  this sequence has a unique upper bound. Indeed, define $s\in \Phi_\Sigma X$ as follows: for every $k\in \N$ there exists $n\in \N$ such that all the trees $s_n, s_{n+1}, s_{n+2}, \dots$ agree up to height $k$. Then this is how $s$ is defined up to height $k$.

It is easy  to verify that $s$ is a well-defined $\Sigma$-tree over $X$. This is obviously an upper bound: to verify $s_m<  s$ for every $m$, one shows, for the height $k$ of the finite tree $s_m$,
 that $s_m$ and $s$ agree at that height, hence $s_m = \partial _k s$. Every other upper bound $s'$ agrees with $s$ on heights $0,1,2,\dots$ -- thus, $s=s'$.

\vskip 2mm
(c) Given a directed set $A \subseteq \Psi_\Sigma X$, all strictly increasing $\omega$-chains in $A$ have the same upper bound. Indeed, let $(s_n)$ and $(s'_n)$ be strictly increasing sequences in $A$, then since $A$ is directed, we can find a strictly increasing sequence $(s''_n)$ in $A$  such that each  $s''_n$ is an upper bound of $s_n$ and $s'_n$ for every $n$. The unique upper bound of that sequence is also an upper bound for $(s_n)$ and $(s'_n)$.
\end{remark}

\begin{corollary}\label{C:psi}
$\Psi_\Sigma X$ is a cpo, i.e., it has directed joins.
\end{corollary}

Indeed, if a directed set $A\subseteq \Psi_\Sigma X$ has a largest element, then this is $\sqcup A$. Assuming the contrary, we can find a strictly increasing sequence $s_n \in A$. If $s$ is its upper bound, then $s=\sqcup A$. In fact, given $x\in A$, we can find a strictly increasing sequence $s'_n \geq s_n$ in $A$ with $x\leq s'_0$ (since $A$ is directed). Since $\sqcup s'_n$ is an upper bound of $(s_n)$, it follows that $\sqcup s'_n =s$. Thus, $s$ in an upper bound of $A$, and it is clearly the smallest one.

\begin{remark}\label{R:ideal}
(1) A monotone function between posets is called \textit{continuous} if it preserves all existing directed joins.

\vskip 1mm
 (2) Recall that a \textit{conservative completion} of a poset $P$  is
a cpo $\bar P$ containing $P$ as a subposet closed under existing directed joins with the following universal property:
 \begin{enumerate}
\item[]   For every continuous  function $f\colon P\to Q$,     where $Q$ is 
 a cpo, there exists a unique continuous extension $\bar f\colon \bar P \to Q$.
 \end{enumerate}
See \cite{BN}, Corollary 2, for the proof that $\bar P$ exists. 

\vskip 1mm
(3)
$\Psi_\Sigma X$ is a conservative completion  of $\Phi_\Sigma X$. Indeed, given a continuous function  $f\colon \Phi_\Sigma X \to Q$, define $\bar f \colon \Psi_\Sigma X \to Q$ by $\bar f(s)=\underset{n\in \N}{\sqcup} f(\partial_ns)$ for every tree $s$ in $\Psi_\Sigma X$. This extends $f$, and the proof of Corollary~\ref{C:psi} demonstrates that $\bar f$ is continuous. It is unique: from $s=\sqcup \partial_n s$ the formula  for $\bar f$ follows via continuity.
\end{remark}

\section{The limit $F^\omega 1$ as a cpo}\label{sec3}

In this section $F$ denotes a finitary set functor with $F\emptyset \ne \emptyset$. If we choose an element $p\colon 1\to F\emptyset$, then the limit $F^\omega =\lim\limits_{n\in\N} F^n 1$ of the terminal-coalgebra chain carries a structure of a cpo (a poset with directed  joins). This cpo was presented in \cite{A}, we recall this structure here and show in the next section a more intuitive description of that cpo ordering.

\begin{notation}\label{N:chain}

(1)  The initial algebra is denoted by $\mu F$ with the algebra structure $\varphi\colon F(\mu F) \to F$. The terminal coalgebra is denoted by $\nu F$ with the structure $\tau\colon \nu F\to F(\nu F)$.

(2) For the initial object $0$ (empty set)
 the unique morphism $i\colon 0\to F0$ yields an $\omega$-sequence of objects $F^n0$ ($n\in \N$) and connecting
morphisms $F^n i$ called the \textit{initial-algebra $\omega$-chain}. Its colimit is denoted by $F^\omega 0$ with the colimit cocone $i_n\colon F^n 0 \to F^\omega 0$. Since $F$ is finitary, $F^\omega 0$ is an initial algebra. The algebra structure $\varphi \colon  F(F^\omega 0)\to F^\omega 0$ is   the unique morphism with  $\varphi \cdot Fi_n = i_{n+1}$ for $n\in \N$. See \cite{A1}.

\vskip 1mm
(3) Dually, the unique morphism $t\colon F1 \to 1$ yields an $\omega^{\op}$-sequence of objects $F^n 1 $ ($n\in \N$)  and connecting morphisms $F^nt$, called the \textit{terminal coalgebra $\omega$-chain}. Its  limit is denoted by $F^\omega 1$ with the limit cone $t_n\colon F^\omega 1 \to F^n 1$.

\vskip 1mm
(4) The unique morphism $u\colon0 \to 1$ defines morphisms $F^n u \colon F^n0 \to F^n 1$. There exists a unique monomorphism $\bar u \colon F^\omega 0 \to F^\omega 1$ with $t_n \cdot \bar u \cdot i_n = F^n u$ ($n\in \N$), see \cite[Lemma 2.4]{A}.

\vskip 1mm
(5) Since $p\colon 1 \to F0$ has been chosen, we get morphisms
\begin{align*}
e_n &=\bar u \cdot i_{n+1} \cdot F^n p\colon F^n 1 \to F^\omega 0\,,\\
\intertext{and we define}
r_n&= e_n \cdot t_n \colon F^\omega 1 \to F^\omega 1\,.
\end{align*}

The following theorem is Theorem 3.3 in \cite{A}. The assumption, made in that paper, that $F$ is bicontinuous, was not used in the proof. Observe that the statement concerns the limit $F^\omega 1$ of which we do \textit{not} claim it is $\nu F$. 

\end{notation}

\begin{theorem}\label{T:chain} 

$F^\omega 1$ is a cpo w.r.t.\ the following  ordering
$$
x \sqsubseteq y \quad \mbox{iff}\quad x=y \quad \mbox{or}\quad x=r_n(y) \quad \mbox{for some $n\in \N $.}
$$
Every strictly increasing $\omega$-chain has a unique upper bound in $F^\omega 1$.
\end{theorem}

\begin{example}\label{E:chain} 
(1) For $F= H_\Sigma$ we have  $F^\omega  1 =\nu H_\Sigma$, all $\Sigma$-trees. Recall our choice of $p \in F0 =\Sigma_0$. The ordering $\sqsubseteq$ above is precisely that by cutting, see Definition~\ref{D:cut}.

Indeed, $\bar u \colon \mu H_\Sigma \to \nu H_\Sigma$ is just the inclusion map. If we put $1=\{p\}$, then $H_\Sigma 1$ consists of $\Sigma$-trees  $\sigma(p, \dots , p)$ or $\sigma \in \Sigma_0$  of height at most $1$ with leaves labelled by $p$. More generally, $H_\Sigma ^n 1$ consists of $\Sigma$-trees of height at most $n$ with leaves of height $n$ labelled by $p$. The function $e_n\colon H_\Sigma ^n 1\to \mu H_\Sigma$ is the inclusion map, hence, $r_n$ is the cutting function $\partial_n$ of Section~\ref{sec2}.

\vskip 1mm
(2) For the finite power-set functor $\cp_f$ we have  $\cp_f0 =\{\emptyset\}$, thus the chosen element is $p=\emptyset$. Recall that a non-ordered tree is called \textit{extensional} if for every node all maximum subtrees are pairwise distinct (i.e., non-isomorphic). Every tree has an \textit{extensional quotient}    obtained  by recursively identifying equal maximum subtrees of every node.

In the initial-algebra chain, $\cp_f^n 0$ can be described as the set of all extensional trees of height at most $n$ (and $\cp_f^n i$ are the inclusion maps). Hence $\cp_f^\omega 0= \bigcup\limits_{n\in\N} \cp_f^n 0$ is the set of all finite extensional trees.

Worrell proved that $\cp^\omega_f 1$ can be described as the set of all compactly branching \textit{strongly extensional} trees, see \cite{W}. (Given a tree $s$,  a relation $R$ on its nodes is called a \textit{tree  bisimulation} if (a) it only relates nodes of the same height and (b) given 
$xRy$, then for every successor $x'$ of $x$ there is a successor $y'$ of $y$  with $x'R y'$, and vice versa. A tree is called  strongly extensional if every tree bisimulation is  contained in the diagonal relation.)
\end{example}

\begin{remark}\label{R:new}
Observe that each $r_n$ factorizes through $\mu F$: we have morphisms
$$
\partial_n \colon \nu F\to \mu F \quad \mbox{with} \quad r_n =\bar u \cdot \partial_n\,.
$$
Indeed, put $\partial_n= i_{n+1} \cdot F^n p\cdot t_n$.

\end{remark}

\begin{notation}[See \cite{A1}]\label{N:not}
The \textit{initial-algebra chain} for $F$ beyond the above finitary iterations is the following chain indexed by all ordinals $n$: on objects define $F^n0$ by $
F^00 =0$, $F^{n+1} 0 = F(F^n0)$ {and} $ F^k 0= \colim\limits_{n<k} F^n0 $
for limit ordinals $k$.
The connecting morphisms are denoted by $i_{n,k} \colon F^n0\to F^k0 $ ($n\leq k$). We have $i_{0,1} \colon 0\to F0$ unique, $i_{n+1, k+1} = Fi_{n,k}$, and for limit ordinals $k$ the cocone $(i_{n,k})_{n<k}$ is a colimit cocone.

Dually, the \textit{terminal-coalgebra chain} indexed by $\Ord^{\op}$ has objects $F^n1$ with $F^0 1 =1$, $F^{n+1} 1= F(F^n1)$ and $F^k 1 =\underset{k>n}{\operatorname{\it lim\ }} F^n 1$. And it has  connecting morphisms $t_{n,k}$ with $t_{1,0}$ unique, $t_{n+1, k+1}= Ft_{n,k}$ and $(t_{n,k})_{k>n}$ the  limit cone if $k$ is a limit ordinal. In our notation above we thus have $t= t_{1,0}$, $Ft=t_{2,1}$, etc.
\end{notation}

\begin{lemma}\label{L:hat}
 Every  natural transformation $\varepsilon \colon H\to F$ between en\-do\-functors induces

{\rm (1)} a unique natural transformation  $\hat \varepsilon_n \colon H^n 1\to F^n 1$ ($n\in \Ord$) between their terminal-coalgebra chains satisfying
 $$
 \hat \varepsilon_{n+1} \equiv H(H^n 1)\xrightarrow{\ \varepsilon_{H^n 1}\ } F(H^n1) \xrightarrow{\ F\hat\varepsilon_n\ }
 F(F^n 1)\,,
$$
and

{\rm (2)} a unique  natural transformation $\tilde\varepsilon_n \colon H^n 0 \to F^n 0$ ($n\in \Ord$) between their initial-algebra chains satisfying
$$
\tilde \varepsilon_{n+1} \equiv H(H^n 0) \xrightarrow{\ \varepsilon_{H^n 0}\ } F(H^n  0) \xrightarrow{\ F \tilde\varepsilon_n\ } F(F^n 0)\,.
$$ 

 \end{lemma}

\begin{proof}
We present the proof of (1), that of (2) is completely analogous.

Denote by $t_{n,k}$ and $t'_{n,k}$ the connecting morphisms of the terminal-coalgebra chains for $F$ and $H$, resp.

We have $\hat \varepsilon_0 \colon 1\to 1$ unique, and $\hat\varepsilon_1 =\varepsilon_1 \colon H1 \to F1$ is also unique.
The first naturality square
$$
\xymatrix{
H1\ar[r]^{t'_{1,0}} \ar[d]_{\hat\varepsilon_1} & 1 \ar[d]^{\hat\varepsilon_0}\\
F1\ar [r]_{t_{1,0}} & 1
}
$$
trivially commutes.

Given $\hat\varepsilon_n$, then $\hat\varepsilon_{n+1}$ is uniquely determined  by the above formula. And every naturality square for $n$
 $$
\xymatrix{
H^n1\ar[r]^{t'_{n,m}} \ar[d]_{\hat\varepsilon_n} & H^m1 \ar[d]^{\hat\varepsilon_m}\\
F^n1\ar [r]_{t_{n,m}} &F^m 1
}
\quad (m\leq n)
$$
yields the following naturality square for $n+1$:
$$
\xymatrix{
H^{n+1}1\ar[rrr]^{H t'_{n,m}} \ar[dd]_{\hat\varepsilon_{n+1}} \ar[rd]^{\varepsilon_{H^n 1}}
&&&
 H^{m+1} 1 \ar[dd]^{\hat\varepsilon_{m+1}} \ar[ld]_{\varepsilon_{H^m1}}\\
& F(H^n 1) \ar[r]^{F{ t'_{n,m}}} \ar[ld]_{F\hat\varepsilon_n} & \ar[dr]_{F\hat\varepsilon_m}&\\
F^{n+1} 1\ar [rrr]_{ Ft_{n,m}} &&&F^{m+1} 1
}
$$
Indeed, the upper part commutes since $\varepsilon \colon H\to F$ is natural, and for  the lower one apply $F$ to the square above.

Thus, all we need proving is that given a limit  ordinal $k$ for which all the above squares with $m\leq n<k$ commute, there is a unique $\hat \varepsilon_k \colon H^k 1 \to F^k 1$ making the following squares
$$
\xymatrix{
H^k1\ar[r]^{t'_{k,n}} \ar[d]_{\hat\varepsilon_k} & H^n1 \ar[d]^{\hat\varepsilon_n}\\
F^k1\ar [r]_{t_{k,n}} &F^n 1
}
\quad (n<k)
$$
commutative. The morphism $\hat \varepsilon_n \cdot t'_{k,n}$ for all $n<k$ form a cone  of the $k$-chain with limit $F^k 1$, i.e., we  have, for each $n>m$, the following commutative triangle
$$
\xymatrix@C=1pc @R=3pc{
&& H^k 1 \ar[dl]_{t'_{k,n}} \ar[dr]^{t'_{k,m}}&&\\
& H^n 1 \ar[ld]_{\hat\varepsilon_n} \ar@{-->}[rr]_{t'_{n,m}} &&
 H^m1 \ar[dr]^{\hat \varepsilon_m}&\\
 F^n 1 \ar[rrrr]_{t_{n,m}} &&&& F^m 1
}
 $$
  Thus, $\hat \varepsilon_k$ is uniquely determined by the above commutative squares.
\end{proof} 

\begin{remark}\label{R:simple} %
$\hat\varepsilon_\omega\colon H^\omega 1 \to F^\omega 1$ is the unique morphism satisfying $\hat \varepsilon_n \cdot t'_n = t_n \cdot \hat\varepsilon_\omega$ for every $n\in \N$. Indeed, this follows from the above proof since $t_n = t_{\omega, n}$ and $t'_n= t'_{\omega,n}$. Analogously, $\tilde \varepsilon_\omega\colon H^\omega 0 \to P^\omega 0$ is the unique morphism satisfying $\tilde \varepsilon \cdot i'_n = i_n\cdot \tilde \varepsilon_n$ for every $n\in \N$.
\end{remark}

\begin{remark}\label{R:W} %
Recall the description of the terminal coalgebra of a finitary set functor $F$ due to Worrell \cite{W}:
\begin{enumerate}
\item[(a)] All connecting morphisms
$t_{n, \omega}$ with $n\geq \omega$
 are monic, thus, $F^{\omega +\omega} 1 = \bigcap\limits_{n\in \N} F^{\omega+ n} 1$;

\item[(b)] $F^{\omega +\omega} 1$ is the terminal coalgebra whose coalgebra structure is inverse to $t_{\omega + \omega +1, \omega +\omega}$.
\end{enumerate}
\end{remark}

\begin{exam}\label{E:W} 
For $\cp_f$ (see  \ref{E:chain}(2)) the subset $\cp^{\omega +n} 1$ of $\cp^{\omega} 1$ consists of all strongly extensional  compactly branching trees which are finitely branching at all levels up to $n-1$. Thus, $\bigcap\limits_{n\in \N} \cp^{\omega +n} 1$ is the set $\nu \cp_f$ of all finitely branching strongly extensional trees in $\cp_f^\omega 1$. This was proved in \cite{W}.
\end{exam}

\begin{remark}\label{R:mono} 
Since $\mu F$ can be viewed as a coalgebra for $F$ (via $\varphi^{-1}$), we have a unique coalgebra homomorphism
$$
m\colon \mu F \to \nu F \quad \mbox{with}\quad \tau \cdot m = Fm \cdot \varphi^{-1}\,.
$$
This is monic for every finitary set functor, see \cite[Proposition 5.1]{A2}.
\end{remark}

We thus can consider $\mu F$ as a subset of $\nu F$ and $m$ as the inclusion map.

Since both $H_\Sigma$ and $F$ are finitary functors, we have the morphism $\tilde \varepsilon_\omega \colon \mu H_\Sigma \to \mu F$ of Lemma~\ref{L:hat}.

\begin{lemma}\label{L:basic}
$\tilde \varepsilon_\omega \colon(\mu H_\Sigma,\varphi') \to (\mu F, \varphi \cdot \varepsilon_{\mu F})$  is a  homomorphism of algebras for $H_\Sigma$. Consequently, $\tilde \varepsilon_\omega$ is a restriction of $\hat k$, i.e., we have 
$\hat k\cdot m'=m\cdot \tilde \varepsilon_\omega\colon \mu H_\Sigma\to \nu F$.
\end{lemma}

\begin{proof}
(1) To verify
that  $\tilde \varepsilon_\omega$ is a homomorphism, i.e.,
 $\tilde \varepsilon_\omega\cdot \varphi' =\varphi \cdot \varepsilon_{\mu F} \cdot H_\Sigma \tilde \varepsilon_\omega$, we use the fact that the colimit cocone $(i'_n)_{n\in \N}$ yields a colimit cocone $(H_\Sigma i'_n)_{n\in \N}$. And each $H_\Sigma i'_n$ merges the two sides of our equation:
\begin{align*}
\tilde \varepsilon_{\omega} \cdot \varphi' \cdot H_\Sigma i'_n &=\tilde \varepsilon_{\omega} \cdot i'_{n+1}
 &&\mbox{(definition of $\varphi'$)}\\[4pt]
&=  i_{n+1} \cdot \tilde \varepsilon_{n+1}
&& \mbox{(definition of $\tilde \varepsilon_{\omega}$)}\\[4pt]
&= \varphi\cdot F i_n \cdot \tilde \varepsilon_{n+1}
 && \mbox{(definition of $\varphi$)}\\[4pt]
&= \varphi \cdot F(i_n\cdot \tilde \varepsilon_n)\cdot \varepsilon_{F^n 0}
&& \mbox{(definition of $\tilde\varepsilon_{n+1}$)}\\[4pt]
&= \varphi \cdot \varepsilon_{\mu F} \cdot H_\Sigma (i_n\cdot \tilde \varepsilon_n)
 && \mbox{($\varepsilon$ natural)}\\[4pt]
 &= \varphi \cdot \varepsilon_{\mu F} \cdot H_\Sigma \tilde \varepsilon_\omega \cdot H_\Sigma i'_n
  && \mbox{(definition of $\tilde\varepsilon_\omega$)}\,.
\end{align*}

(2) We observe that $m$ and $m'$ are homomorphisms of algebras  for $H_\Sigma$. Indeed, $\tau \cdot m= Fm \cdot \varphi^{-1}$ in Remark~\ref{R:mono} yields
$$
m\cdot (\varphi \cdot \varepsilon_{\mu F}) = \tau^{-1} \cdot Fm \cdot \varepsilon_{\mu F} = (\tau^{-1} \cdot \varepsilon_{\nu F}) \cdot H_\Sigma m\,,
$$
 analogously for $m'$. Due to (1) this shows that $ m\cdot \tilde \varepsilon_{\omega} \colon (\mu H_\Sigma, \varphi') \to (\nu F, \tau^{-1} \cdot \varepsilon_{\nu F})$ is a homomorphism for $H_\Sigma$. So is $\hat k\cdot m'$, thus the initiality of $\mu F$ yields $\hat k\cdot m' = m\cdot \tilde \varepsilon_{\omega} $.
\end{proof}


\section{The Order by Cutting}\label{sec4}

We have seen in Section~\ref{sec2} that for  polynomial  functors the terminal coalgebra $\nu H_\Sigma$ is a cpo when ordered by cutting of the $\Sigma$-trees. In the present section we represent  an arbitrary finitary set functor $F$ as a quotient of some $H_\Sigma$. This will enable us to introduce an order by cutting on $\nu F$ and $\mu F$. We then prove the following, whenever $F\emptyset \ne \emptyset$:
\begin{enumerate}
\item[(a)] if $F$ is \textit{bicontinuous}, i.e., preserves also limits of $\omega^{\op}$-chains, then $\nu F$ is a cpo which is the conservative completion  of $\mu F$,
\end{enumerate}
\noindent
and
\begin{enumerate}
\item[(b)] for $F$ in general $\nu F$ and $\mu F$ share the same conservative completion.
\end{enumerate}

\begin{definition}\label{D:pres} 
By a \textit{presentation} of a set functor $F$ is  meant a finitary signature $\Sigma$ and a natural transformation $\varepsilon \colon H_\Sigma 
\twoheadrightarrow F$ with epic components.
\end{definition}

\begin{prop}[See \cite{AT}]\label{P:pres} 
A set functor has a presentation iff it is finitary. The category of algebras for $F$ is then equivalent to a variety of $\Sigma$-algebras.
\end{prop}

\begin{remark}\label{R:proof} 
The proof is not difficult: a possible signature for $F$ is $\Sigma_n = Fn$ for $n\in \N$. Yoneda Lemma yields a natural transformation from $\Sigma_n \times \Set(n, -)$ to $F$ for every $n\in \N$, and this defines  $\varepsilon \colon H_\Sigma \to F$ which is epic iff $F$ is finitary.
\end{remark}

Moreover, if elements of $H_\Sigma X =\coprod\limits_{n\in \N} \Sigma_n \times X^n$ are represented as flat terms $\sigma (x_1, \dots , x_n)$, then we define \textit{$\varepsilon$-equations} as equations of the following form:
 $$
 \sigma (x_1, \dots , x_n) = \tau(y_1, \dots, y_m)
  $$
  such that $\sigma \in \Sigma_n$, $\tau \in \Sigma_m$, and $\varepsilon_X$ merges the given  elements  of $H_\Sigma X$. (Here $X= \{x_1, \dots , x_n, y_1, \dots y_m\}$.) The variety of $\Sigma$-algebras presented by all $\varepsilon$-equations is equivalent to the category of $F$-algebras. This equivalence takes an algebra $\alpha\colon FA \to A$ to the $\Sigma$-algebra $\alpha \cdot \varepsilon_A \colon H_\Sigma A \to A$.

\begin{corollary}\label{C:pres}
The initial algebra $\mu F$ is  the quotient of the algebra $\mu H_\Sigma$ of finite $\Sigma$-trees modulo the congruence $\sim$ merging trees $s$ and $s'$ iff $s$ can be obtained  from $s'$ by a (finite) application of $\varepsilon$-equations.
\end{corollary}

\begin{example}\label{E:pres}
The finite power-set functor $\cp_f$ has  a presentation by the signature $\Sigma$ with a unique $n$-ary operation for every $n\in \N$. Thus, $H_\Sigma X = X^\ast$. And  we consider the natural transformation $\varepsilon_X \colon X^\ast \to \cp_f X$  given by $(x_1\dots x_n)\mapsto \{x_1, \dots , x_n\}$.

$\mu H_\Sigma$ can be described as  the algebra of all (unlabelled) finite trees. And two trees are congruent iff  they have the same extensional quotient, see Example~\ref{E:chain}.
Consequently, $\mu \cp_f = \mu H_\Sigma \big/ \sim$ is the set of all finite unordered extensional trees.
\end{example}

\begin{remark}\label{R:pres}
Analogously to $\mu F= \mu H_\Sigma \big/\sim$ above, we can describe the terminal coalgebra $\nu F$ as a quotient of $\nu H_\Sigma$, whenever a nullary symbol $p\in \Sigma_0$ is chosen, as follows. In \cite[3.13]{AM}, the congruence $\sim^\ast$ on $\nu H_\Sigma$  of a possibly infinite application of $\varepsilon$-equations was defined as follows:
$$
s\sim^\ast  s'\quad \mbox{iff}\quad 
\partial_n s \sim \partial _n s' \ (n\in \N)\,.
$$
\end{remark}

\begin{theorem}[{\cite[3.15]{AM}}]\label{T:term} 
 The quotient coalgebra $\nu H_\Sigma\big/\sim^\ast$ is, when considered as an $F$-coalgebra, the terminal coalgebra. Shortly,
 $$
 \nu F= \nu H_\Sigma \Big/\sim^\ast\,.
 $$
\end{theorem}

 \begin{remark}\label{R:q}

Let $\tau'\colon \nu H_\Sigma \to H_\Sigma(\nu H_\Sigma)$ and $\tau \colon \nu F \to F(\nu H)$ denote the respective coalgebra structures.
 The quotient map $\hat k \colon \nu H_\Sigma \to \nu F$ is a homomorphism of coalgebras for $F$, i.e., 
the following square
$$
\xymatrix{
\nu H_\Sigma \ar[d]_{\hat k} \ar[r]^{\tau'} & H_\Sigma (\nu H_\Sigma)
\ar[r]^{\varepsilon_{\nu H_\Sigma}}& F(\nu H_\Sigma)\ar[d]^{F \hat k}\\
\nu F \ar[rr]_{\tau} && F (\nu F)
}
$$
commutes.
 This was proved in \cite{AM}, see the proof of Theorem 3.15 there (where $\hat k$ was denoted by $\hat \varepsilon$).
 \end{remark}

 \begin{lemma}\label{L:split}
 The morphism $\hat k\colon \nu H_\Sigma \to \nu F$ is a split epimorphism.
 \end{lemma}
 
 \begin{proof}
 Choose $ b\colon F(\nu F) \to H_\Sigma (\nu F)$ with $\varepsilon_{\nu F} \cdot b =\id$. For  the coalgebra $b\cdot \tau \colon \nu F \to H_\Sigma (\nu F)$ we have a unique homomorphism $k^\ast \colon \nu F \to \nu H_\Sigma$ with $\tau' \cdot k^\ast = H_\Sigma k^\ast \cdot (b\cdot \tau)$. We prove $\hat k\cdot k^\ast =\id$ by verifying that $\hat k \cdot k^\ast$ is an endomorphism of the terminal coalgebra $\nu F$, i.e., $\tau \cdot (\hat k \cdot k^\ast)=F(\hat k\cdot k^\ast)\cdot \tau$:
 \begin{align*}
 \tau\cdot \hat k\cdot k^\ast &= F \hat k\cdot \varepsilon_{\nu H_\Sigma} \cdot \tau' \cdot k^\ast&& \mbox{($\hat k$ a homomorphism)}\\
 &= F k\cdot \varepsilon_{\nu H_\Sigma} \cdot H_\Sigma k^\ast \cdot b\cdot \tau && \mbox{($k^\ast$ a homomorphism)}\\
&= F(k\cdot k^\ast)\cdot \varepsilon_{\nu F} \cdot b\cdot \tau && \mbox{($\varepsilon$ natural)}\\
&= F(k\cdot k^\ast)\cdot \tau && \mbox{($\varepsilon_{\nu F} \cdot b=\id$)}
\end{align*}
\end{proof}
 
\begin{definition}\label{D:term} %
The following relation $\leq$ on $\nu F$ is called \textit{order by cutting}: given distinct  congruence classes $[s] $ and $ [s']$ of $\sim^\ast$, put 
$$
[s] < [s'] \quad \mbox{iff} \quad s\sim \partial_n s'\quad \mbox{for some}\quad n\in \N\,.
$$
\end{definition}

We obtain posets $\nu F$ and $\mu F$ (as a subposet via $\bar u$ see Remark~\ref{R:mono}).

\begin{example}\label{E:term} %
For the presentation of $\cp_f$ of Example~\ref{E:pres} we know that $\nu H_\Sigma$ is the algebra of all  finitely branching trees. We have $s\sim^\ast s'$ iff the extensional quotients of $\partial_n s$ and $\partial_ns'$ coincide for all $n\in \N$. This way Barr described $\nu \cp_f$ in \cite{B}.

Consequently, 
for extensional trees we have 
$s < s'$ iff $s$ is the extensional quotient of some cutting of $s'$.
\end{example}

\begin{notation}\label{N:term} 
In the rest of the present section we assume that $F$ is a finitary set functor with $F\emptyset \ne \emptyset$, and that a presentation $\varepsilon$ is given. Since $\varepsilon_{\emptyset}\colon \Sigma_0 \to F\emptyset$ is epic, we can choose a nullary symbol $p'$ in $\Sigma_0$. This yields a choice of $p=\varepsilon_\emptyset (p')$ in $F\emptyset$.

We use the notation $\tau$, $\varphi$,  $r_n$ etc. for $F$ as in  Section~3, and the corresponding notation $\tau'$, $\varphi'$, $r'_n$ etc.\ for $H_\Sigma$. Recall $\hat \varepsilon_\omega \colon \nu H_\Sigma \to F^\omega 1$ from Lemma~\ref{L:hat}.

\end{notation}

\begin{remark}\label{R:least}
(1) The homomorphism $\hat k\colon \nu H_\Sigma \to \nu F$ of Remark~\ref{R:q} is clearly monotone and  preserves the least elements. Indeed, if $p'\in \Sigma_0$ is the chosen element, then the least element of $\nu H_\Sigma$ is the singleton tree labelled by $p'$. And the least element of $\nu F$ is $[p']=\hat k(p')$.

(2) Since $\tilde \varepsilon_\omega$ is a domain-codomain restriction of $\hat k$, see Lemma~\ref{L:basic},
 it also is  monotone and  preserves the least element.
\end{remark}

\begin{prop}\label{P:ideal}
The morphisms $r_n \colon F^\omega 1\to F^\omega 1$ and $r'_n \colon H_\Sigma^\omega 1 \to H_\Sigma^\omega 1$ are related by $r_n \cdot \hat \varepsilon_\omega = \hat \varepsilon_\omega \cdot r'_n$ ($n\in \N$).
\end{prop}

\begin{proof}
(1) 
We prove $F^n u\cdot \tilde \varepsilon_n = \hat \varepsilon_n \cdot H_\Sigma^n u$ by induction on $n\in \N$. The first step is trivial. The induction step is computed as follows:
\begin{align*}
F^{n+1} u\cdot \tilde \varepsilon_{n+1} &= F(F^n u\cdot \tilde \varepsilon_n)\cdot \varepsilon_{H_\Sigma ^n0} &&\mbox{(definition of $\tilde\varepsilon_n$)}\\[4pt]
&= F(\hat \varepsilon_n \cdot H_{\Sigma}^n u) \cdot \varepsilon_{H_\Sigma^n0} && \mbox{(induction hypothesis)}\\[4pt]
&= F \hat \varepsilon_n \cdot \varepsilon_{H_\Sigma^n 1} \cdot H_\Sigma ^{n+1} u && \mbox{($\varepsilon$ natural)}\\[4pt]
&= \hat \varepsilon_{n+1} \cdot H_\Sigma^{n+1}u && \mbox{(definition of $\hat \varepsilon_n$)}\,.
\end{align*}

(2)
We next verify $\bar u \cdot \tilde \varepsilon_\omega= \hat \varepsilon_\omega \cdot \bar u'$.
 For that it is sufficient to prove, for all $n\in \N$, that $\bar u \cdot \tilde \varepsilon_\omega \cdot i'_n= \hat \varepsilon_\omega \cdot \bar u' \cdot i'_n$. Indeed, $(i'_n)$ is a collectively epic cocone. Thus, we only need to verify, by induction on $k \in \N$, that $t_{n+k}$ merges the two sides of that equation: $t_{n+k} \cdot (\bar u \cdot \tilde \varepsilon_\omega\cdot i'_n) = t_{n+k} \cdot (\hat \varepsilon_\omega \cdot \bar u' \cdot i'_n)$. (Here we use the fact that $(t_{n+k})_{k\in \N}$ is a collectively monic cone for every $n$.)
 
 This follows for $k=0$ from the following computation:
 \begin{align*}
t_n\cdot \bar u \cdot \tilde \varepsilon_\omega \cdot i'_n &= t_n \cdot \bar u \cdot i_n \cdot \tilde \varepsilon_n && \mbox{see Remark~\ref{R:simple}}
\\
& = F^n u\cdot \tilde \varepsilon_n && \mbox{(definition of $\bar u$)}\\
& = \hat \varepsilon_n \cdot H_\Sigma^n u && \mbox{see (1)}\\
&= \hat \varepsilon_n\cdot t'_n \cdot \bar u' \cdot i'_n && \mbox{(definition of $\bar u'$)}\\
&= t_n \cdot \hat\varepsilon_\omega \cdot \bar u'\cdot i'_n &&\mbox{see Remark~\ref{R:simple}}\,.
\end{align*}
And if the above equation holds for $k$, then we  can write $t_{n+(k+1)}$ as $t_{(n+1)+k}$ and apply the above equation  to $k$ and $n+1$. From that we obtain the induction step:
\begin{align*}
t_{n+(k+1)}\cdot \bar u \cdot \tilde \varepsilon_\omega \cdot i'_n &= t_{(n+1)+k} \cdot \bar u \cdot \tilde \varepsilon_\omega \cdot i'_{n+1} \cdot H_\Sigma^n i && \mbox{($i'_n$ compatible)}\\
&= t_{(n+1)+k} \cdot \hat \varepsilon_\omega \cdot \bar u' \cdot i'_{n+1} \cdot H_\Sigma^n i && \mbox{(induction hypothesis)}\\
&= t_{n+(k+1)} \cdot \hat \varepsilon_\omega \cdot \bar u'\cdot i'_n && \mbox{($i'_n$ compatible)}\,.
\end{align*}

(3) Now we prove   for the given point $p=\varepsilon_\emptyset \cdot p'\colon 1\to F0$ that $F^n p\cdot \hat \varepsilon_n=\tilde \varepsilon_{n+1}\cdot H_\Sigma^n p'$. This is trivial for $n=0$, and the induction step is as follows:
\begin{align*}
F^{n+1} p \cdot \hat \varepsilon_{n+1} &= F^{n+1} p\cdot F\hat \varepsilon_n\cdot \varepsilon_{H^n1}&& \mbox{(definition of $\hat \varepsilon_n$)}\\
&= F \big(\tilde\varepsilon_{n+1} \cdot H_\Sigma^n p'\big)\cdot \varepsilon_{H^n1} && \mbox{(induction hypothesis)}\\
&= F\tilde \varepsilon_{n+1} \cdot \varepsilon_{H_\Sigma^{n+1}1}\cdot H_\Sigma^{n+1} p' && \mbox{($\varepsilon$ natural)}\\
 &= \tilde \varepsilon_{n+2} \cdot H_\Sigma^{n+1} p' && \mbox{(definition of $\tilde\varepsilon_n$)}\,.
 \end{align*}
 
 (4) The proof of our proposition follows. Recall that $r_n$ is defined by
 $$
 r_n = e_n \cdot t_n = \bar u\cdot i_{n+1} \cdot F^n p\cdot t_n
 $$ 
and analogously $r_n'$. Thus
\begin{align*}
r_n\cdot \hat \varepsilon_\omega &= \bar u\cdot i_{n+1} \cdot F^n p\cdot t_n \cdot \hat \varepsilon_\omega &&\\
&= \bar u \cdot i_{n+1} \cdot F^np\cdot \hat \varepsilon_n\cdot t'_n &&
\mbox{see Remark~\ref{R:simple}}\\
&= \bar u\cdot i_{n+1}  \cdot \tilde \varepsilon_{n+1}\cdot H_\Sigma^n p'\cdot t_n' && \mbox{see (3)}\\
&= \bar u  \cdot \tilde\varepsilon_\omega \cdot i'_{n+1}  \cdot H_\Sigma^n p' \cdot t'_n && \mbox{see Remark~\ref{R:simple}}\\
&= \hat \varepsilon_\omega \cdot \bar u'\cdot i'_{n+1} \cdot H_\Sigma^n p' \cdot t'_n && \mbox{see (2)}\\
&= \hat \varepsilon_\omega \cdot r'_n\,.
\end{align*}
\end{proof}

In the following theorem $\mu F$ is considered as a subset of $\nu F$ via the monomorphism $m$, see Remark~\ref{R:mono}. Thus $(\mu F)^A$, ordered component-wise, is a subposet of $(\nu F)^A$.
Moreover, $F(\mu F)$ is considered as a poset via the bijection $\varphi$, and analogously for $F(\nu F)$.

\begin{theorem}\label{T:good}
Let $F$ be a finitary set functor with $F\emptyset \ne \emptyset$. 
 The order of $\nu F$ by cutting coincides with that of Theorem  \ref{T:chain}.
And 
 the poset $\nu F$ has the same conservative completion  as its subposet $\mu F$.
The coalgebra structure $\tau$ is the unique continuous extension of $\varphi^{-1}$.

\end{theorem}
\begin{proof}
(1) Recall that $t_n=t_{\omega,n}$ and $Ft_n= t_{\omega+1, n+1}$, thus
$$
t_{n+1}\cdot t_{\omega +\omega, \omega} = t_{\omega +\omega, n+1} = Ft_n \cdot t_{\omega +\omega, \omega+1}\,.
$$
Moreover, observe that since $\tau^{-1} = t_{\omega +\omega +1, \omega+\omega}$, we have $t_{\omega +\omega, \omega+1} \cdot \tau^{-1} = F t_{\omega +\omega, \omega}$.

\vskip 1mm
(2) We  prove that the homomorphism $\hat k\colon \nu H_\Sigma \to \nu F$ of Remark~\ref{R:q}
fulfils
$$
\hat\varepsilon_\omega = t_{\omega+\omega, \omega} \cdot \hat k \colon \nu H_\Sigma \to F^\omega 1\,.
$$
Following Remark~\ref{R:simple}  we need to prove the following equalities
$$
t_n\cdot (t_{\omega+\omega, \omega}\cdot \hat k) =\hat \varepsilon_n  \cdot t_n' \qquad (n\in \N)\,.
$$
The case $n=0$ is trivial. The induction step is as follows:
\begin{align*}
\hat \varepsilon_{n+1} \cdot t'_{n+1} &= F\hat \varepsilon_n \cdot \varepsilon_{H_\Sigma^n 1} \cdot t'_{n+1} && \mbox{(definition of $\hat \varepsilon_n$)}\\
&= F\hat \varepsilon_n \cdot \varepsilon_{H_\Sigma^n 1} \cdot H_\Sigma t'_n \cdot \tau' && \mbox{($\tau' = (t'_{\omega+1, \omega})^{-1}$}\\
& && \mbox{and $t'_n = t'_{\omega,n}$)}\\
&= F(\hat\varepsilon_n\cdot t'_n)\cdot \varepsilon_{\nu H_\Sigma} \cdot \tau' && \mbox{($\varepsilon$ natural)}\\
&= Ft_n \cdot Ft_{\omega + \omega, \omega} \cdot F\hat k \cdot \varepsilon_{\nu H_\Sigma} \cdot \tau' && \mbox{(induction hypothesis)}\\
&= Ft_n \cdot t_{\omega + \omega, \omega+1}\cdot \tau^{-1} \cdot F\hat k \cdot \varepsilon_{\nu H_\Sigma} \cdot \tau ' && \mbox{by (1)}\\
&=  Ft_n \cdot t_{\omega + \omega, \omega+1}\cdot \hat k && \mbox{($\hat k$ a homomorphism)}\\
&= t_{n+1} \cdot t_{\omega +\omega, \omega} \cdot \hat k && \mbox{by (1)}\,.
\end{align*}

(3) The congruence $\sim^\ast$ is the kernel equivalence of $\hat k$, see Remark~\ref{R:q}. Since $t_{\omega+\omega, \omega}$ is monic, it follows from (1) that this is also the kernel equivalence of $\hat\varepsilon_\omega$.

\vskip 1mm
(4) The ordering of $F^\omega 1$ defined in Theorem \ref{T:chain} coincides, when restricted to $\nu F$ (via the embedding $t_{\omega +\omega, \omega}$), with the ordering by cutting. To prove this, we verify that, given elements $x=[t]$ and $y=[s]$ of $\nu F$, the following equivalence holds for every $n\in \N$:
$$
t\sim^\ast \partial_n s \quad \mbox{iff} \quad t_{\omega +\omega, \omega} (x) = r_n\cdot t_{\omega +\omega, \omega} (y)\,.
$$
That is, we are to prove for all $n\in \N$ that
$$
\hat \varepsilon_\omega(t) = \hat \varepsilon_\omega \cdot r'_n(s) \quad \mbox{iff}\quad t_{\omega +\omega, \omega}(x) = r_n \cdot t_{\omega +\omega, \omega} (y)\,.
$$
Due to (2), this translates to the following equivalence
$$
\hat  \varepsilon_\omega (t) =\hat  \varepsilon_\omega \cdot r'_n(s) \quad \mbox{iff}\quad \hat  \varepsilon_\omega (t) = r_n \cdot \hat  \varepsilon_\omega (s)\,,
$$
which follows from Proposition~\ref{P:ideal}.

\vskip 1mm
(5) For the morphism $\bar u \colon \mu F \to F^\omega 1$ of  3.1(4) we prove that
$$
\bar u= t_{\omega +1, \omega}\cdot F\bar u \cdot \varphi^{-1}\,.
$$
It is sufficient to prove that the equality holds when precomposed by $i_{n+1}\colon F^{n+1} 0 \to \mu F$ for every $n\in \N$. Since $i_{n+1} = 
i_{n+1, \omega}$ and $\varphi^{-1} = i_{\omega, \omega+1}$, we have $\varphi^{-1} \cdot i_n = i_{n+1, \omega+1}= F i_n$. Thus we want to verify
$$
\bar u\cdot i_{n+1} = t_{\omega +1,\omega} \cdot F(\bar u \cdot i_n)\,.
$$
For that, we postcompose by $t_{n+1+k}$ for all $k\in \N$ (and use that given any $n$  this cone is collectively monic):
$$
t_{n+1+k}\cdot \bar u \cdot i_{n+1} = t_{\omega+1, n+1+k} \cdot F(\bar u\cdot i_n)\,.
$$
This equation holds for $k=0$ since the left-hand side is $F^{n+1}u$, see 3.1(4), and the right-hand one is
$$
t_{\omega +1, n+1} \cdot F(\bar u \cdot i_n) = Ft_n \cdot F\bar u \cdot F i_n = F(F^n u)\,.
$$
The induction step from $k$ to $k+1$ (for $n$ arbitrary) is  easy: just re-write $n+1+k+1$ as $n+2+k$ and use the induction hypothesis on $n+1$ in place of $n$.

\vskip 1mm

\vskip 1mm
(6) For every element $x \in F^\omega 1$, all elements $r_n(x)$ are compact. That is,
given a directed set $D\subseteq F^\omega 1$, then $r_n(x) \sqsubseteq
\bigsqcup D$ implies $r_n(x) \sqsubseteq y$  for some $y\in D$. This
clearly holds for $F=H_\Sigma$. Due to Proposition \ref{P:ideal} and (4)
above, it also follows for $F$.

\vskip 1mm
(7) $F^\omega 1$ is the conservative completion  of $\mu F$. More
precisely, we prove that  the embedding $\bar u\colon \mu F \to F^\omega
1$  has the universal property w.r.t.\  to continuous maps from $\mu F$ to
cpo's. (Observe that  $\mu F$
is trivially  closed under existing directed joins due to Theorem \ref{T:chain}.) 

First, observe that the image of each $r_n$ is a subset of the image of $\bar u$, see Remark~\ref{R:new}.


Every element $x \in F^\omega 1$ yields a sequence $r_n(x)$ in $\mu F$, and  for the order of Theorem~\ref{T:chain} we clearly have $x=\underset{n\in \N}{\bigsqcup} r_n(x)$. Given a monotone function $f\colon \mu F \to B$ where $B$ is a cpo, we define $\bar f\colon F^\omega 1\to B$ by $\bar f(x) =\underset{n\in \N}{\bigsqcup} f\big(r_n(x)\big)$. This is a continuous function. Indeed, given a directed set $D \subseteq F^\omega 1$ we know from Theorem~\ref{T:chain} that $x=\bigsqcup D$ exists. 
Then $D$ is mutually cofinal with $\{r_n(x); n\in \N\}$. This is clear if $x\in D$. Otherwise, (6) implies that
 each $r_n(x)$ is, due to $r_n(x)\sqsubseteq x$, under some element of
$D$. And for each $y\in D$ the fact that $y\sqsubseteq x$ implies that 
 we have $n$ with $y=r_n(x)$. Consequently, $f[D]$ is mutually cofinal with $\{f\big(r_n(x)\big)\}$in $B$, thus, $f(\sqcup D) = f(x) = \sqcup f[D]$.

\vskip 1mm
(8) We  prove that $\bar u$ factorizes through $t_{\omega +\omega, \omega }=\underset{n\in\N}{\bigcap} t_{\omega+n, \omega}$, see Remark~\ref{R:W}.
We verify by induction a factorization through $t_{\omega +n, \omega}$.
 For $n=1$, see (5). For $n=2$ we apply (5) twice: since $Ft_{\omega +1, \omega } = t_{\omega +2, \omega +1}$, we get
\begin{align*}
\bar u&= t_{\omega  +1, \omega } \cdot F (t_{\omega +1, \omega}\cdot F\bar u \cdot \varphi^{-1})\\
&= t_{\omega+2, \omega }\cdot F(F\bar u \cdot \varphi^{-1})\,.
\end{align*}
Analogously for $n=3,4,\dots$.

\vskip 1mm
(9) The proof of  the theorem follows. First, $F^\omega 1$ is the conservative completion  of $\nu F$, the argument is as in (7).
  It follows that $\varphi^{-1} \colon \mu F \to F(\mu F)$, which is a poset  isomorphism (by our definition of the order of $F(\mu F)$) has at most one continuous extension to $\nu F$. And $\tau$ is continuous (indeed, a poset isomorphism, too). Thus, we just need proving that $\tau$ extends $\varphi^{-1}$. In other words, the inclusion map $m$ of Remark~\ref{R:mono} fulfils $\tau\cdot m = Fm\cdot \varphi^{-1}$, and $Fm$ is also the inclusion map.

The latter is clear in case $F$ preserves inclusion  maps.
Next let $F$ be arbitrary. By Theorem III.4.5 in \cite{AT} there exists a set functor $\bar F$ preserving inclusion which agrees with $F$ on all nonempty sets and  functions and fulfils $\bar F \emptyset\ne \emptyset$ (since $F \emptyset\ne \emptyset$).  Then the categories of algebras for $F$ and $\bar F$ also coincide, thus $\mu F= \mu\bar F$. And the categories of nonempty coalgebras for $F$ and $\bar F$ also  coincide, hence, $\nu F = \nu \bar F$. Since the theorem  holds for $\bar F$, it also holds for $ F$.
\end{proof}

\begin{corollary}\label{C:good}
A  bicontinuous set functor with $F\emptyset \ne \emptyset$ has a terminal coalgebra which is the  conservative completion   of its initial algebra. Its coalgebra structure is the unique continuous extension of the inverted algebra structure of $\mu F$.
\end{corollary}

This follows from the above  proof: we have seen that the
 conservative completion  of $\mu F$ is $F^\omega 1$ which, for $F$ bicontinuous, is $\nu F$.

\begin{remark}\label{R:good}
In the proof of the above theorem we have seen that $\hat k$ is a domain restriction of $\hat \varepsilon_\omega$: we have $\hat \varepsilon_\omega = t_{\omega +\omega, \omega} \cdot \hat k$. And the homomorphism $\tilde \varepsilon_\omega$ is a domain-codomain restriction of $\hat k$, see Lemma~\ref{L:basic}.  Consequently, Proposition~\ref{P:ideal} yields
$$
\partial_n\cdot \hat k =\tilde \varepsilon_\omega\cdot \partial_n' \colon \nu H_\Sigma \to \mu F \quad (n\in \N)\,.
$$
Here $\partial_n'$ is the domain-restriction of $r_n'$ and $\partial_n$ that of $r_n$, see Remark~\ref{R:new}.

\end{remark}

\section{Free Iterative Algebras}\label{sec5}

\begin{assumpt} Throughout  this section $F$ is a finitary set functor with a given presentation $\varepsilon \colon H_\Sigma \twoheadrightarrow F$, see Definition~\ref{D:pres}.

\end{assumpt}

\begin{remark} 
Let $X$ be a nonempty set.

(1)
 The initial algebra of $F(-)+X$ is precisely the free algebra for $F$ on $X$: notation $\Phi X=\mu F(-)+X$. Indeed, the components of the algebra structure $\varphi\colon F(\Phi X)+X\to \Phi X$ yield an algebra $\Phi X$ for $F$ and a morphism $\eta \colon X \to \Phi X$, respectively. That $F$-algebra clearly  has  the universal property w.r.t.\ $\eta$.

(2)
 Let us  choose an element $p'\in \Sigma_0 +X$. The finitary functor $F(-) +X$ has the following presentation: the signature is $\Sigma_X $ of Remark~\ref{R:psi}. And the natural transformation can, since $H_{\Sigma_X} = H_\Sigma (-) +X$, be chosen to be
$$
\varepsilon +\id_X \colon H_{\Sigma_X} \twoheadrightarrow F(-) +X\,.
$$
This yields an element $p\in F\emptyset +X$ which is $\varepsilon_\emptyset (p')$ in case $p' \in \Sigma_0$, else $p'=p$.
\end{remark}

\begin{notation} 
$\Phi X$ denotes the poset forming the free algebra on $X$ for $F$ ordered by cutting w.r.t $\varepsilon +\id_X$. And $\sim$ is the congruence on $\Phi_\Sigma X$ (the algebra of finite $\Sigma$-trees on $X$) of applying $\varepsilon$-equations, see Corollary~\ref{C:pres}.
\end{notation}

\begin{remark} 
 We do not speak about $(\varepsilon +\id_X)$-equations, since we do not have to: the function
 $$
 \varepsilon_Z +\id_X \colon H_\Sigma Z +X \to FZ+X
 $$
 does not merge flat terms with variables from $X$, hence, every $(\varepsilon+\id_X)$-equation is simply an $\varepsilon$-equation.
 \end{remark}
 
 \begin{corollary} 
 Free algebras for $F$ are free $\Sigma$-algebras modulo $\varepsilon$-equations: $\Phi X = \Phi_\Sigma X/\sim$.
 \end{corollary}
 
 \begin{exs}\label{E:P} 
 (1) For $F=\Id$ we choose $p\in X$ and obtain $\Phi X=\N\times X$ ordered as follows:
 $$
 (n,x)<(m,y) \quad \mbox{iff}\quad n<m \quad \mbox{and}\quad x=p\,.
 $$
 
 \vskip 1mm
 (2) The functor $FX = A\times X$ yields
$$
\Phi X= A^\ast \times X
$$
ordered by
$$
(u,x)< (v,y) \quad \mbox{iff\quad $u$ is a prefix of $v$ and \ \ $x=p$}.
$$

\vskip 1mm
(3) The functor $FX = X^{I} \times \{0,1\}$ (corresponding to deterministic automata with a finite  input set $I$) is naturally equivalent to $H_\Sigma$, where $\Sigma$ consists of two operations $a$, $b$ of arity $n=\card I$. Thus $\Phi X$ is the algebra of all finite $n$-ary trees with inner nodes labelled by $\{a,b\}$ and leaves labelled in $X$. The order is by tree cutting.

\vskip 1mm
(4) Let $\cp_k$ denote the subfunctor of the power-set functor given by all subsets of at most $k$ elements. We can describe $\Phi X$ as the algebra of all non-ordered, finite  extensional $k$-branching trees (i.e.\ every node has at most $k$ children) with leaves labelled in $X+\{p\}$.

Here we use a signature $\Sigma$ having, for every $n\leq k$, precisely one $n$-ary operation; the nullary one is  called $p$. Then $\Phi_\Sigma X$ is the algebra of all $k$--branching finite trees with leaves labelled in $X +\{p\}$. It is ordered  by tree cutting. And  given $k$-branching trees $s$ and $s'$ we have $s\sim s'$ iff they have the same  extensional quotient, see \ref{E:chain}(2).  This yields the above description of $\Phi X$.

To describe the order of $\Phi X$, let us call the extensional quotient of $\partial_n s$ (cutting $s$ at height $n$) the $n$-th \textit{extensional cutting}. Then for distinct $s$, $s'$ in $\Phi X$ we have  $s<s'$ iff $s$ is an extensional cutting of $s'$.
\end{exs}

\begin{remark}\label{R:M}
Whereas  the initial algebra for $F(-)+X$ is the free algebra for $F$, the terminal coalgebra
$$
\Psi X= \nu F(-)+X
$$
is the free completely iterative algebra for $F$, as we recall below. The concept of a \textit{recursive equation} in an algebra $\alpha\colon FA \to A$ is given by a set $X$ of recursive variables and a morphism $e\colon X\to FX +A$.
\end{remark}

\begin{definition}\label{D:M}
A \textit{solution} of recursive equation $e\colon X\to FX+A$ in an algebra $(A, \alpha)$ is a morphism
$e^\dagger \colon X\to A$ making the following square
$$
\xymatrix{
X\ar[r]^{e^\dagger} \ar[d]_{e} & A\\
FX +A \ar [r]_{ Fe^\dagger +\id} & FA +A \ar[u]_{[\alpha, \id]}
}
$$
commutative. The algebra $(A, \alpha)$ is called \textit{completely iterative} if every  recursive
equation has a unique solution.
\end{definition}

\begin{exam}\label{E:M}
If $F = H_\Sigma$, we can think of $e$ as a system of recursive equations of the form 
$$
x=\sigma (x_1, \dots , x_n) \quad \mbox{or} \quad x=a \quad (a\in A),
$$
one for every variable $x\in X$ (depending on $e(x)$ lying in the left-hand or right-hand summand of $H_\Sigma X+A$). And then the solution $e^\dagger$ makes an assignment of elements of $A$ to variables from $X$ satisfying those recursive equations: from $x =\sigma (x_1, \dots, x_n)$ we get $e^\dagger(x) =\sigma_A \big( e^\dagger(x_1), \dots , e^\dagger(x_n)\big)$, and from $x =a$ we get $e^\dagger(x)=a$.
\end{exam}

The algebra $\nu H_\Sigma$ of $\Sigma$-trees (with the algebra structure $\tau^{-1}$ of tree-tupling) is completely iterative. For every recursive equation $e\colon X\to H_\Sigma X + \nu H_\Sigma$ the solution $e^{\dagger}\colon X \to \nu H_\Sigma$ can be defined as follows: given $n\in \N$ we describe the cut trees $\partial_n' e^\dagger (x)$ for all variables $x\in X$ simultaneously by induction on $n\in \N$:

(1) $ \partial_0' e^\dagger (x)$ is the singleton tree labelled by $p$.

 (2) Given  $\partial_n' e^\dagger (x)$ for all $x\in X$, then for every $x\in X$ with $e(x)=\sigma (x_1,\dots , x_n)$ in the left-hand summand $H_\Sigma X$  we define $\partial_{n+1}' e^\dagger (x)$ to be the tree with root labelled by $\sigma$ and with $n$ subtrees $\partial_n' e^\dagger(x_i)$, $i=1, \dots , n$. Whereas if $e(x)= s\in \nu H_\Sigma$, then $\partial_{n+1}' e^\dagger (x) =\partial_{n+1}' s$.

\begin{theorem}[See \cite{M}]\label{T:M}
Let $\tau_X \colon \Psi X \to F(\Psi X)+X$ be the terminal coalgebra for $F(-) +X$. The components of $\tau_X^{-1}$ make $\Psi X$ an $F$-algebra with a morphism $\eta\colon X\to \Psi X$. This is the free completely iterative algebra for $F$ w.r.t. the universal morphism $\eta$.

In particular, $(\nu F, \tau^{-1})$ is the initial completely iterative algebra.
\end{theorem}

\begin{notation} 
$\Psi X$ denotes the poset forming the free completely iterative algebra on $X$ for $F$ ordered by cutting w.r.t. $\varepsilon +\id_X$. And $\sim^\ast$ is the  congruence on $\Psi_\Sigma X$ (the algebra of $\Sigma$ trees over $X$) of a possibly infinite application of $\varepsilon$-equations, see Remark~\ref{R:pres}.
\end{notation}

\begin{corollary} 
Let $F$ be a bicontinuous set functor. The free completely iterative algebra $\Psi X$  on a set $X\ne \emptyset$ is a cpo which is the 
conservative completion  of the free algebra $\Phi X$.

The algebra structure of $\Psi X$ is the unique continuous extension of the algebra structure of $\Phi X$.
 \end{corollary}
 
 This is an application of Corollary~\ref{C:good} to $F(-) +X$.
 
 \begin{example} 
 (1) For $F=\Id$ the conservative completion  of $\Phi X= \N \times X$ adds just one maximum element as $\underset{n\in\N}{\bigsqcup} (n,p)$. Thus $\Psi X =\N\times X+1$. 
 
 \vskip 1mm
 (2) For $FX = A\times X$ the conservative completion  of $A^\ast \times X$ adds joins to all sequences $(u_0, p)<(u_1, p) <(u_2, p) < \cdots$ where each  $u_n$ is a prefix of $u_{n+1}$ ($n\in \N$). That join is expressed by the infinite word in $A^\omega$ whose  prefixes are all $u_n$. We thus get
 $$
 \Psi X = A^\ast \times X+A^\omega\,.
 $$
 
 \vskip 1mm  
 (3) For the bicontinuous functor $F=\cp_k$ we can describe $\Psi X$ as the algebra of all extensional $k$-branching trees with leaves labelled in $X +\{p\}$.
 
 Indeed, this algebra with the order by  extensional cutting (see Example~\ref{E:P}), is the completion of its subalgebra $\Phi X$ of finite trees. To see this, observe that every strictly increasing sequence $s_0 < s_1<s_2 \cdots$ in $\Psi X$ has a unique upper bound: the tree $s$ defined  level by level so that, given $n$, its extensional cutting at $n$ is the same as that of $s_k$ for all  but finitely many $k\in \N$. Therefore, given a continuous  function $f\colon \Phi X \to B$ where $B$ is a cpo, the unique continuous extension $\bar f\colon \Psi X \to B$ is given by $\bar f(s) = \underset{n\in\N}{\bigsqcup} f(s_n)$ where $s_n$ is the extensional cutting of $s$ at level $n$.

(4) For the functor $FX = X^I \times \{0,1\}$ we have  $\Psi X= n$-ary trees with leaves labelled in $X$ and inner nodes  labelled in $\{a,b\}$. The order is by cutting.

\vskip 1mm
(5) Aczel and Mendler introduced in \cite{AMe} the functor $(-)_2^3$ defined by
$$
X_2^3 =\big\{(x_1, x_2, x_3) \in X^3; x_i=x_j \ \mbox{for some\ } i\ne j\big\}\,.
$$
This is a bicontinuous functor with a presentation using $\Sigma =\{\sigma_1, \sigma_2, \sigma_3\}$, all operations binary, and the following $\varepsilon$-equations
$$
\sigma_1(x,x)= \sigma_2(x,x)= \sigma_3(x,x)\,.
$$
Here $\varepsilon \colon H_\Sigma \to (-)_2^3$ is given by $\sigma_1(x,y)\mapsto (x,x,y)$, $\sigma_2(x,y) \mapsto  (x,y,y)$ and $\sigma_3(x,y)\mapsto  (x,y,x)$.

The free algebra $\Phi X = \Phi_\Sigma X/\sim$ is described  as follows:
$\Phi_\Sigma X$ consists of finite binary trees with leaves labelled in $X$ and inner nodes labelled in $\Sigma$. And $s\sim s'$ means that we can obtain $s$ from $s'$ be relabelling arbitrarily inner nodes whose left and right child yield the same tree. The order is by cutting.

The free completely iterative algebra is $\Psi X = \Psi_\Sigma x\big/\sim^\ast$, where $\Psi_\Sigma X$ are binary trees with leaves labelled in $X$ and  inner nodes  labelled in $\Sigma$. And $\sim^\ast$ allows infinite relabelling of the type above.  $\Psi X$ is a cpo which is the conservative completion  of $\Phi X$.
\end{example}

\begin{corollary} 
For every finitary set functor the free algebra on a set $X\ne \emptyset$ has the same conservative completion  as the iterative algebra on $X$.
The algebra structure of $\Psi X$ is, again, the unique continuous extension of the  algebra structure of $\Phi X$.
\end{corollary}

This is an application of  Theorem~\ref{T:good} to $F(-)+X$.

\begin{example} 
For the finite power-set functor $\cp_f$ the algebra $\Psi  X$ can be described as the quotient $\Psi_\Sigma X\big/\sim^\ast$, where $\Psi_\Sigma X$ are the finitely branching trees with leaves labelled in $X + \{p\}$. And $s\sim^\ast s'$ means  that the extensional cuttings of $s$ and $s'$ are the same for every level $n$.

A better description: $\Psi X$ is the set of all finitely branching strongly extensional trees (see Example~\ref{E:chain}(2)), with leaves labelled in $X + \{p\}$. The proof is completely analogous to that for $\nu \cp_f$ in Worrell's paper \cite{W}.

$\Psi X$ is ordered by extensional cutting. This is not a cpo. To see this, consider  an arbitrary strongly extensional tree $s$ which is not finitely branching. Thus, $s\notin \Psi X$. Each extensional cutting is finite (since for every $n$ we only have a finite number of extensional trees of height $n$) and  this yields an increasing $\omega$-sequence in $\Phi X$ that has no join in $\Psi X$.

The common conservative completion  of $\Phi X$ and $\Psi X$ is the algebra of all compactly branching  strongly extensional trees with leaves labelled in $X+\{p\}$. The proof is, again, analogous to that for $\nu \cp_f$ in \cite{W}.
\end{example}

\section{Approximate Solutions} 

In this section we prove that  solutions of iterative equations in free iterative algebras are obtainable as joins of $\omega$-chains of approximate solutions. This is true for  every finitary set functor $F$ and every nonempty set of recursion variables. We first prove the corresponding result for the terminal coalgebra considered as an algebra $\tau^{-1} \colon F(\nu F) \to \nu F$.

Throughout this section a presentation $\varepsilon \colon H_\Sigma \to F$ is assumed and a choice of an element $p\in F\emptyset +X$ where $X$ is a fixed set of ``recursion'' variables. In particular at the begining we set $X=\emptyset$ and choose $p\in F\emptyset$, i.e., we work with a finitary functor with $F\emptyset \ne \emptyset$.

We continue to use $\tau$, $\varphi$, $\partial_n, \dots $ for $F$ and $\tau'$, $\varphi'$, $\partial_n', \dots$  for $H_\Sigma$ (as in Section~\ref{sec4}). We know that $\nu F$ is the initial completely iterative algebra. We are going to describe solutions $e^\dagger \colon X \to \nu F$ of recursive equations $e\colon X \to FX + \nu F$ as joins of $\omega$-chains
$$
e^\dagger_0 \sqsubseteq e^\dagger_1 \sqsubseteq e^\dagger_2 \dots \colon X\to \mu F
$$
of approximate solutions in the initial algebra. Here we work with the poset $(\nu F)^X$ ordered pointwise and its subposet $(\mu F)^X$.

Recall $\partial_n$ from Remark~\ref{R:new}.

\begin{definition}\label{D:appr}
The $k$-th \textit{approximate solution} $e^\dagger_k \colon X\to \nu F$ of a recursive equation $e\colon X \to FX + \nu F$ is defined by induction on $k\in \N$ as follows:
$$
e^\dagger_0 \colon X \to \mu F \quad \mbox{is the least element of the poset $(\mu F)^X$,}
$$
and given $e^\dagger_k$, then the following square defines $e^\dagger_{k+1}$:
$$
\xymatrix@C=5pc{ 
X \ar[r]^{e^\dagger_{k+1}} \ar[d]_{e} & \mu F\\
FX + \nu F\ar[d]_{\id+\partial_k} &\\
FX + \mu F \ar[r]_{F e^\dagger_k + \id} & F(\mu F)+ \mu F \ar[uu]_{[\varphi, \id]}
}
$$
\end{definition}

 We are going to prove that the unique solution $e^\dagger$ of $e$ in $\nu F$ is the join of the $\omega$-chain  $e^\dagger_k$  considered in $(\nu F)^X$. Or, more precisely, for the inclusion $m\colon \mu F \to \nu F$ of Remark~\ref{R:mono} we have 
 $$
 e^\dagger = \bigsqcup\limits_{k\in\N} m \cdot e^\dagger_k\,.
 $$

 \begin{example}\label{E:pol}
 If $F=H_\Sigma$ then $e_n^\dagger$ is precisely the cutting $\partial_n' e^\dagger$ of Example~\ref{E:M}. This is obvious for $n=0$, and the induction step is easy.
 \end{example}

\begin{theorem}
Let $F$ be a finitary set functor with $F\emptyset \ne \emptyset$. For every recursive equation $e\colon X\to FX + \nu F$ the unique solution in $\nu F$ is the join of the $\omega$-chain of approximate solutions $e^\dagger_n$ $(n\in \N)$ in the poset $(\nu F)^X$.
\end{theorem}

\begin{proof}
We know from Example~\ref{E:M} that the theorem holds for $H_\Sigma$. We apply this to the following  recursive equation  w.r.t. $H_\Sigma$:
$$
e' \equiv X \xrightarrow{\ \ e\ \ } FX +\nu F \xrightarrow{\ \ b+k^\ast\ \ }H_\Sigma X + \nu H_\Sigma
$$
where $b$ is a splitting  of $\varepsilon_X$ and $k^\ast$ splits $\hat k$, see Lemma~\ref{L:split}. Thus, $e = (\varepsilon_X + \hat k)\cdot e'$. We know that $(e')^\dagger$ is the join $(e')^\dagger = \underset{n\in \N}{\bigsqcup} m' \cdot (e')^\dagger$ for the inclusion $m' \colon \mu H_\Sigma \to \nu H_\Sigma$. From that we derive $e^\dagger = \underset{n\in \N}{\bigsqcup} m \cdot e^\dagger_n$ by proving  that (1) $e^\dagger =\hat k \cdot (e')^\dagger$ and (2) $e^\dagger_n =\tilde \varepsilon_\omega  \cdot (e')^\dagger_n$ for $n\in \N$ (see Remark~\ref{R:simple}). Indeed, 
we then have
$$
e^\dagger =\hat k \cdot \underset{n\in \N}{\bigsqcup} m' \cdot  (e')^\dagger_n = \underset{n\in \N}{\bigsqcup} \hat k \cdot m' \cdot (e')^\dagger_n
$$
since post-composition with $\hat k$ preserves  the order and all joins that exist in $(\mu H_\Sigma)^X$: recall
from Remark~\ref{R:least} 
 that $\hat k \colon \nu H_\Sigma \to (\nu  H_\Sigma)\big/\sim^\ast$ is the quotient map inducing the order by cutting on $\nu F$. We get from Lemma~\ref{L:basic}
$$
e^\dagger =\underset{n\in \N}{\bigsqcup} m\cdot \tilde \varepsilon_\omega\cdot  (e')^\dagger_n = \underset{n\in \N}{\bigsqcup}  m \cdot e^\dagger_n
$$
as required.

(1) Proof of 
$e^\dagger= \hat k \cdot (e')^\dagger$. It is sufficient to prove that $\hat k\cdot (e')^\dagger$ solves $e$ in the algebra $(\nu F, \tau^{-1})$, i.e., it  is equal to $[\tau^{-1}, \id] \cdot \big( F[\hat k\cdot (e')^\dagger]+\id\big)\cdot e$.
This follows from the commutative diagram below, since $e=(\varepsilon_X +\hat k) \cdot e'$:
$$
\xymatrix@C=2pc{ 
X \ar[r]^{(e')^\dagger} \ar[d]_{e'}& \nu H_\Sigma \ar[rr]^{\hat k} && \nu F\\
H_\Sigma X\!\! +\! \nu H_\Sigma \ar[r]^<{\quad  H_\Sigma (e')^\dagger\! +\id}  
 \ar[dr]_{H_\Sigma[\hat k\cdot (e')^\dagger]+\id} \ar[dd]_{\varepsilon_X +\hat k}     
& H_\Sigma (\nu H_\Sigma)\! +\! \nu H_\Sigma \ar[r]^{ \varepsilon_{\nu H_\Sigma} +\id}\ar[u]_{[(\tau')^{-1}, \id]}  \ar[d]^{H_\Sigma \hat k+\id}  & F(\nu H_\Sigma) \!+ \nu H_\Sigma\ar@{}[d]|{(N)} \ar[ddr]^{F\hat k +\hat k}&\\
&H_\Sigma(\nu F)\! +\!\nu H_\Sigma \ar[drr]^{\varepsilon_{\nu F} +\hat k} &&\\
FX\! +\!\nu F \ar@{}[ur]|{(N)}\ar[rrr]_{F[\hat k \cdot (e')^\dagger ] +\id}&&& F(\nu F) \!+ \!\nu F \ar[uuu]_{[\tau^{-1}, \id]}
}
$$
The upper left-hand  part expresses that $(e')^\dagger$ solves $e'$. For all the other inner parts consider the components of the corresponding coproducts separately. The right-hand components commute in each case trivially. The left-hand components of the parts denoted by $(N)$ commute since $\varepsilon$ is natural. For the upper right-hand part recall that $\hat k$ is a homomorphism, i.e., $\tau \cdot \hat k = F  \hat k\cdot \varepsilon_{\nu H_\Sigma}$.

(2) The proof of  $e^\dagger_n =\tilde k \cdot (e')^\dagger_n$ is performed by induction on $n\in N$. The case $n=0$ is trivial since $\tilde \varepsilon_\omega$ preserves the least element (see Remark~\ref{R:least}) and $(e')^\dagger_0$ is the constant map of that value. The induction step follows from the commutative diagram below:
$$
\xymatrix@C=2pc{ 
& X \ar[dl]_{e} \ar[d]_{e'}\ar[r]^{(e')^\dagger_{n+1}}&
   \mu H_\Sigma \ar[r]^{\tilde \varepsilon_\omega} &  \mu F\\
   FX\! +\!\nu F \ar[ddd]_{\id +\partial_n} & H_\Sigma X \!+ \!\nu H_\Sigma \ar[l]_{\varepsilon_X + \hat k} \ar[d]_{\id +\partial_n'}&&\\
   & H_\Sigma X \!\!+\!\mu H_\Sigma  \ar[ddl]_{\varepsilon_X +\tilde \varepsilon_\omega} \ar[r]_{H_\Sigma(e')^\dagger_n +\id\ \ \  }&  H_\Sigma(\mu H_\Sigma)\!+\! \mu H_\Sigma \ar[uu]_{[\varphi', \id]} \ar[d]^{\varepsilon_{\mu H_\Sigma}+\tilde \varepsilon_\omega}&\\
   && F(\mu H_\Sigma) \!+ \!\mu F \ar[dr]^{F\tilde \varepsilon_\omega +\id}&\\
   FX \!+ \!\mu F \ar[rru]^{F(e')^\dagger_n +\id} \ar[rrr]_{F e^\dagger_n +\id} &&& F(\mu F)\! +\! \mu F \ar[uuuu]_{[\varphi, \id]}
   }
   $$
The lower triangle commutes since $e^\dagger_n =\tilde \varepsilon_\omega \cdot (e')^\dagger_n$ by induction  hypothesis. The upper square is the definition of $(e')^\dagger_{n+1}$ and the part right of it commutes due to $\tilde\varepsilon_\omega$ being a homomorphism, see Lemma~\ref{L:basic}. The middle part commutes by naturality of $\varepsilon$. For the lower left-hand part 
see Remark~\ref{R:good}.

(3)
It remain 
$s\bigsqcup (e_n^\dagger)$
to verify that $(e_n)_{n\in \N}$ 
is an $\omega$-chain in $(\nu F)^X$. For $(e')_n^\dagger$ this follows from $(e')^\dagger_n=\partial_n \cdot (e')^\dagger$, see Example~\ref{E:M}. Thus, we only need to observe that $(e')_n^\dagger\leq (e')_{n+1}^\dagger$ implies $\tilde \varepsilon_\omega \cdot (e')_n^\dagger \leq \tilde \varepsilon_\omega\cdot (e')_{n+1}^\dagger$.
 Indeed, see Remark~\ref{R:least}.
\end{proof}

\begin{definition}\label{D:hom}
For every coalgebra $\alpha\colon X \to FX$ we define \textit{approximate homomorphisms} $h_n\colon X\to \mu F$ by induction on $n\in \N$ as follows: $h_0$ is the least element of $(\mu F)^X$, and given $h_n$ we put
$$
h_{n+1} \equiv X \xrightarrow{\ \alpha \ } FX \xrightarrow{\ Fh_n\ } F(\mu F) \xrightarrow{\ \varphi\ }\mu F\,.
$$
\end{definition}

\begin{corollary}\label{P:hom}
For every coalgebra $(X, \alpha)$ the unique homomorphism to $\nu F$ is the join of the $\omega$-chain of approximate homomorphisms in $(\nu F)^X$.
\end{corollary}

\begin{proof}
Let $h\colon (X, \alpha) \to (\nu F, \tau)$ be the unique homomorphism. Form the recursive equation
$$
e\equiv X  \xrightarrow{\ \alpha \ } FX  \xrightarrow{\ \operatorname{inl} \ } FX + \nu F\,.
$$
Then $e^\dagger_n= h_n$ for every $n\in \N$. This is clear for $n=0$. The induction step follows from the square in Definition~\ref{D:appr}:
observe that $(\id +\partial_n)\cdot e= (\id +\partial_n)\cdot \operatorname{inl}\cdot \alpha = \operatorname{inl}\cdot \alpha$.

Moreover, $h$ is a solution of $e$: from $\tau \cdot h=Fh\cdot \alpha$ we get $h= \tau^{-1} \cdot Fh\cdot \alpha = \tau^{-1}\cdot (Fh+\id) \cdot \operatorname{inl}\cdot \alpha$, as required. Thus, our corollary follows from the preceding theorem.
\end{proof}


\begin{corollary}\label{C:appr}
Let $F$  be a finitary set functor. For every nonempty set $Y$ the solutions of recursive equations in the free iterative algebra $\Psi Y$ are obtained as joins of  $\omega$-chains of the approximate solutions in the free algebra $\Phi Y$.
\end{corollary}

This is just an application of  the above theorem to the functor $F(-)+Y$ and a choice $p\in Y$ making $\Psi Y$ a poset by cutting.

\section{Conclusions and Open problems}\label{sec6} 

Terminal coalgebras of finitary set functors $F$ carry a canonical partial order which is a cpo whenever $F$ is bicontinuous. This was observed  by the author a long time ago.
  The present paper describes this order in a completely new manner, using the cutting of $\Sigma$-trees for a signature $\Sigma$ presenting $F$. In the bicontinuous case the terminal coalgebra is the conservative completion  of the initial algebra of $F$. Moreover the algebra structure of $\mu F$ determines the coalgebra structure of $\nu F$ as the unique continuous extension of the inverted map.
  
  The above results are applied to free completely iterative algebras $\Psi X$ for $F$ on all nonempty sets $X$.
  In the bicontinuous case $\Psi X$ is the conservative completion  of the free algebra $\Phi X$ on $X$, and the algebra structure of $\Psi X$ is the unique continuous extension of that of $\Phi X$. For finitary set functors in general, $\Phi X$ and $\Psi X$ have the same conservative  completions. We have demonstrated this on several examples of ``everyday'' finitary functors.
  Our main result is that solutions of recursive equations in $\Psi X$ can be obtained as joins of $\omega$-chains of (canonically defined) approximate solutions in $\Phi X$.
  
  It is an open problem whether an analogous result can be proved for accessible set functors in general. Another important question is whether there is a reasonable class of locally finitely presentable categories such that a similar order of free iterative algebras can be presented for every finitary endofunctor.

\bibliographystyle{plainurl}

\nocite{*}
\bibliography{iterative}

\end{document}